\documentclass[letterpaper, 10 pt, conference]{ieeeconf}  

\IEEEoverridecommandlockouts                              

\usepackage{graphics} 
\usepackage{amsmath} 
\usepackage{amssymb}  

\usepackage[hang,flushmargin]{footmisc} 

\newcommand{\R}{\mathbb R}
\newcommand{\dfb}{\stackrel{\Delta}{=}}

\DeclareMathOperator{\image}{image}

\newtheorem{theorem}{Theorem}
\newtheorem{proposition}{Proposition}
\newtheorem{corollary}{Corollary}
\newtheorem{lemma}{Lemma}

\newcommand{\vertiii}[1]{{\left\vert\kern-0.25ex\left\vert\kern-0.25ex\left\vert #1 
    \right\vert\kern-0.25ex\right\vert\kern-0.25ex\right\vert}}
\def\pip{\ \rule[-.035in]{.031in}{.153in}}
\newcommand{\widebars}[1]{\pip\hspace{-.01in} #1 \hspace{-.02in}\pip}

\newcommand{\gensn}[1]{|#1|} 
\newcommand{\indsn}[1]{|#1|} 
\newcommand{\procsn}[1]{\widebars{#1}} 
\newcommand{\restnorm}[1]{\|#1\|_R} 
\DeclareMathOperator{\coe}{coe} 

\newcommand{\consubsp}{\mathcal{I}} 
\newcommand{\eqrowsum}{\mathcal{M}} 

\title{\LARGE \bf
Consensus Seminorms and their Applications}

\author{Ron Ofir, Ji Liu, A. Stephen Morse, and Brian D. O. Anderson
\thanks{*The work of the first three authors was supported in part by the Air Force Office of Scientific Research, under award numbers FA9550-23-1-0175 and FA9550-25-1-0223. The work of R. Ofir was partially supported by the Viterbi Fellowship, Technion. The work of J. Liu was supported in part by the National Science Foundation under grant 2230101.}
\thanks{R. Ofir and A. S. Morse  are with the Department of Electrical Engineering, Yale University, CT, USA ({\tt\small \{ron.ofir,as.morse\}@yale.edu}).
J. Liu is with the Department of Electrical and Computer Engineering at Stony Brook University, NY, USA ({\tt\small ji.liu@stonybrook.edu}).
B. D. O. Anderson is with the School of Engineering, Australian National University, Canberra, Australia ({\tt\small brian.anderson@anu.edu.au}).}
}

\begin{document}

\maketitle
\thispagestyle{empty}
\pagestyle{empty}

\begin{abstract}

Consensus is a well-studied problem in distributed sensing, computation and control, yet deriving useful and easily computable bounds on the rate of convergence to consensus remains a challenge. This paper discusses the use of seminorms for this goal. A previously suggested family of seminorms is revisited, and an error made in their original presentation is corrected, where it was claimed that the a certain seminorm is equal to the well-known coefficient of ergodicity. Next, a wider family of seminorms is introduced, and it is shown that contraction in any of these seminorms guarantees convergence at an exponential rate of infinite products of matrices, generalizing known results on stochastic matrices to the class of matrices whose row sums are all equal one. Finally, it is shown that such seminorms cannot be used to bound the rate of convergence of classes larger than the well-known class of scrambling matrices.

\end{abstract}

\section{INTRODUCTION}

Various problems in distributed control and computation can be formulated either as consensus problems~\cite{OlfatiSaber2007ConsensusAndCooperation} or as problems in which consensus plays a key role \cite{Nedic2009DistributedSubgradient,Mou2015DistributedLinSolve}. In all cases, a fundamental question is deciding whether or not a given sequence of neighbor sets and weights will lead to all agents reaching a consensus, and if so, determining at what rate will convergence occur. If the agents are known to reach consensus, then for any initial condition the states of the agents will be closer to consensus than they were at initial time after some number of steps. One approach to bounding the rate of convergence is to look sufficiently far ahead in time and then determine the rate using the geometric average over that horizon, see for example~\cite{Olshevsky2011ConvergenceSpeed} for the time-invariant case and~\cite{Ming2008ConsensusRate} for the time-varying case. 

A different approach is to look for a measure in which contraction occurs at every step of the process~\cite{Xiao2003FastLinearIterations, DePasquale2024seminorms, Ipsen2011Ergodicity, Liu2011DetGossip}. Since the states do not converge to the origin but rather to some finite asymptotic state which depends on the initial values, contraction is more naturally studied with respect to a seminorm rather than a norm. This approach was used in~\cite{Xiao2003FastLinearIterations} to bound and optimize the rate of convergence in a distributed averaging process, although the term ``seminorm'' was not mentioned explicitly. Reference~\cite{Ipsen2011Ergodicity} defines a larger family of seminorms and uses them to bound the rate of convergence in the contexts of nonhomogeneous Markov chains and rank-one updates to matrices.

To be useful for determining a bound on the rate of convergence, it is particularly important to find a seminorm which can be evaluated in a straightforward manner. Indeed,~\cite{Ipsen2011Ergodicity} provides such expressions for some of the seminorms defined there. A different family of easily computable seminorms was defined and studied in~\cite{Liu2011DetGossip}. More recently, the idea of \emph{inducing} a matrix seminorm by a vector seminorm was studied in~\cite{DePasquale2024seminorms}, and explicit formulas for the induced seminorms were derived in both discrete and continuous time contexts. The goal of this paper is to further pursue this approach of bounding the rate of reaching consensus using seminorms.

We begin by defining a general family of suitable seminorms which we call consensus seminorms. We show that these include the seminorms of~\cite{Liu2011DetGossip} as a special case, and then show that one of the explicit expressions derived in~\cite{Liu2011DetGossip} is unfortunately incorrect. Next, we discuss the idea of inducing a matrix seminorm using a different definition than the one in~\cite{DePasquale2024seminorms}, and provide an elementary proof showing that the induced seminorm of a certain vector seminorm is the well-known coefficient of ergodicity. Focusing back on general consensus seminorms, we show that contraction in any submultiplicative consensus seminorm implies convergence to consensus at an exponential rate. Finally, we show that there is no class of stochastic matrices larger than the class of scrambling matrices, where all matrices in the class are contracting in a single submultiplicative consensus seminorm.

\section{Consensus seminorms}\label{sec:consensus_seminorms}

Recall that a seminorm satisfies the same properties as a norm, see e.g. \cite{HornJohnson2012MatrixAnalysis}, with the exception that a seminorm being zero does not imply that its argument is necessarily zero. In the context of consensus problems, a seminorm being zero should be equivalent to the system reaching consensus, which is in turn equivalent to the matrix product~$M_i M_{i-1} \cdots M_1$ converging as $i \to \infty$ to a rank one matrix of the form~$\mathbf{1} c$ for some~$c \in \R^{1 \times n}$. Therefore, this paper considers seminorms that are zero for all matrices whose rows are all the same.

While it may seem natural to define a consensus seminorm over all of~$\R^{n \times m}$, it will be made clear in Section~\ref{sec:induced} in considering induced consensus seminorms why this is impossible to do in a consistent manner in general. Furthermore, while matrices related to consensus problems are typically stochastic, i.e., have only nonnegative entries and such that each row sums to 1, seminorms must be defined over a linear space and stochastic matrices are clearly not closed under addition or scalar multiplication. For these reasons, we define~$\eqrowsum^{n \times m}$ as the linear space of real~$n \times m$ matrices whose row sums are all equal, and consider seminorms over this space. To be exact, we call a seminorm~$\gensn{\cdot} : \eqrowsum^{n \times m} \to \R$ a~\emph{consensus seminorm} if~$\gensn{M} = 0$ if and only if $M = \mathbf{1} c$ for some~$c \in \R^{1 \times m}$, where~$\mathbf{1}$ denotes the vector all of whose entries equal 1. As with matrix norms, a useful property which some but not all consensus seminorms have is the \emph{submultiplicative property}, namely that
\begin{equation}
    \gensn{M_2 M_1} \le \gensn{M_2} \gensn{M_1}
\end{equation}
for any~$M_1,M_2 \in \eqrowsum^{n \times n}$.

Another useful property that \emph{all} consensus seminorms satisfy is \emph{shift invariance}. Fix~$M \in \eqrowsum^{n \times m}$ and~$c \in \R^{1 \times m}$, then by the triangle inequality and the fact that~$\gensn{\mathbf{1}c} = 0$,
\begin{align*}
    \gensn{M} = \gensn{M + \mathbf{1}c - \mathbf{1}c} \le \gensn{M + \mathbf{1}c} + \gensn{-\mathbf{1}c} = \gensn{M + \mathbf{1}c},
\end{align*}
and
\begin{equation*}
    \gensn{M + \mathbf{1}c} \le \gensn{M} + \gensn{\mathbf{1}c} = \gensn{M}.
\end{equation*}
We conclude that for any consensus seminorm, and any~$M \in \eqrowsum^{n \times m}$ and~$c \in \R^{1 \times m}$,
\begin{equation}\label{eq:shift_invariance}
    \gensn{M + \mathbf{1}c} = \gensn{M}.
\end{equation}

\section{Metric consensus seminorms}

In this section we discuss the \emph{metric consensus seminorm}~$\procsn{\cdot}_p : \R^{n \times m} \to \R, \; p\in[1,\infty],$ defined by\footnote{Unlike general consensus seminorms, the metric consensus seminorm can be defined over~$\R^{n \times m}$ rather than just~$\eqrowsum^{n \times m}$.}
\begin{equation}\label{eq:procsn_def}
    \procsn{M}_p = \min_{c \in \R^{1 \times m}} \|M - \mathbf{1}c\|_p,
\end{equation}
where~$\|M\|_p$ denotes the standard matrix norm induced by the $p$-norm. The name ``metric'' follows from the relation of this seminorm to the so-called~\emph{metric projection} of~$M$ to the nearest element in the space of rank one matrices whose rows are all the same.

This family of seminorms was first studied in~\cite{Liu2011DetGossip}. It is obvious that the metric consensus seminorm is indeed a consensus seminorm, that is~$\procsn{M}_p = 0$ if and only if~$M = \mathbf{1}c$ for some~$c \in \R^{1 \times m}$. Metric consensus seminorms are also submultiplicative\footnote{Although the proof of this property in \cite{Liu2011DetGossip} establishes~\eqref{eq:metric_submul} only for square matrices $M_1$ and $M_2$ with the row sums of $M_2$ all equal $1$, essentially the same proof applies to the more general case when the matrices are not square and the row sums of $M_2$ are equal, but not necessarily equal to $1$.}: for any~$M_1 \in \R^{m \times k}$ and~$M_2 \in \eqrowsum^{n \times m}$,
\begin{equation}\label{eq:metric_submul}
    \procsn{M_2 M_1}_p \le \procsn{M_2}_p \procsn{M_1}_p.
\end{equation}

Most importantly, metric consensus seminorms have easy to evaluate explicit expressions for the cases~$p=1,2$~\cite{Liu2011DetGossip}. Let~$\mathbf{n} \dfb \{1,2,\dots,n\}$.
\begin{proposition}
    Let~$M \in \R^{n \times n}$, and let~$q$ denote the unique integer quotient of~$m$ divided by 2. Then,
    \begin{equation}
        \procsn{M}_1 = \max_{j \in \mathbf{m}} \left\{ \sum_{i \in \mathcal{L}_j} m_{ij} - \sum_{i \in \mathcal{S}_j} m_{ij} \right\}
    \end{equation}
    where~$\mathcal{L}_j$ and~$\mathcal{S}_j$ denote the~$q$ largest and smallest entries of column~$j$, respectively.
\end{proposition}
\begin{proposition}
    Let~$M \in \R^{n \times n}$, and let~$P = I - \frac{1}{n}\mathbf{1}\mathbf{1}'$ denote the orthogonal projection onto the image of~$\mathbf{1}\mathbf{1}'$. Then
    \begin{equation}
        \procsn{M}_2 = \|PM\|_2 = \mu\left\{ M' P' P M \right\}
    \end{equation}
    where~$\mu\{T\}$ denotes the largest eigenvalue of~$T$.

    Furthermore, if~$M$ is doubly stochastic then~$\procsn{M}_2 = \sigma_2(M)$ where~$\sigma_2$ denotes the second largest singular value of~$M$.
\end{proposition}

For~$p=\infty$ it was also claimed in~\cite{Liu2011DetGossip} that the seminorm evaluates to the well-known \emph{coefficient of ergodicity}, also known as the \emph{Dobrushin coefficient}, defined as
\begin{equation}\label{eq:coe}
    \coe(M) = \frac{1}{2} \max_{i,j} \sum_{k=1}^n |m_{ik} - m_{jk}|,
\end{equation}
where~$m_{ij}$ denotes the $ij$th entry of~$M \in \R^{n \times m}$. It is not difficult to show that for any stochastic matrix~$S$,~$\coe(S) \le 1$, and furthermore~$\coe(S) < 1$ if and only if no two rows of~$S$ are orthogonal~\cite{Seneta2006}. Because of their role in non-homogeneous Markov chains, stochastic matrices which have no two rows orthogonal to each other are also known as \emph{scrambling matrices}. It was recently discovered by us~\cite{DeterministicCorrection2025} that the metric consensus seminorm with~$p = \infty$ is \emph{not equal} to the coefficient of ergodicity, and in the sequel we provide a counterexample.

Verifying the counterexample requires the following necessary and sufficient condition, guaranteeing that~$\procsn{S}_\infty < 1$ in terms of the $n\times n$ matrix $[[S]]$ whose $ij$th entry is $1$ if $s_{ij}\neq 0$ and is $0$ otherwise.

\begin{proposition}\label{prop:procsn_infty_contract}
    Let~$S \in \R^{n \times n}$ be a stochastic matrix. Then,~$\procsn{S}_\infty < 1$ if and only if there exists a nonnegative vector~$y$ such that
    \begin{equation}\label{eq:proc_infinity_semicontract}
        (\mathbf{1}\mathbf{1}' - 2[[S]])y < \mathbf{0},
    \end{equation}
    where the inequality is entrywise and~$\mathbf{0} \in \R^n$ denotes the zero vector.
\end{proposition}

\begin{proof}
    By~\eqref{eq:procsn_def} and since~$S$ is stochastic,~$\procsn{S}_\infty \le \|S\|_\infty = 1$. Since~$\|S - \mathbf{1}c\|_\infty$ is convex in~$c$, we have that~$\procsn{S}_\infty < 1$ if and only if~$c'=\mathbf{0}$ is not a local minimum of~$\|S - \mathbf{1}c\|_\infty$. That is,~$\procsn{S}_\infty < 1$ if and only if for any for any positive constant $\epsilon$, no matter how small, there is a nonzero row vector $\delta\in\R^{1\times n}$ with $\|\delta\|_{\infty} \leq \epsilon$ such that $\|S-\mathbf{1}\delta\|_{\infty}<1$.
    Let $\epsilon$ be the smallest nonzero entry in~$S$ thereby ensuring that,
    \begin{equation*}
         s_{ij} - \delta _j \geq 0, \;\;\;\; j \in \mathcal{P}_i, \;\;\;\;\; i \in \mathbf{n}
    \end{equation*}
    where $\mathbf{n} = \{1,2,\ldots,n\}$, $s_{ij}$ is the $ij$th entry in $S$, $\delta _j$ is the $j$th entry in $\delta$, and $\mathcal{P}_i$ is the set of indices $j\in\mathbf{n}$ for which $s_{ij}$ is positive.
    Then for $i\in\mathbf{n}$,
    \begin{equation*}
        \sum_{j=1}^n|s_{ij} - \delta _j| = \sum_{j\in\mathcal{P}_i}(s_{ij}-\delta_j) + \sum_{j\in\bar{\mathcal{P}}_i}|\delta_j|,
    \end{equation*}
    where $\bar{\mathcal{P}}_i$ is the complement of $\mathcal{P}_i$ in $\mathbf{n}$.
    Since for $i\in\mathbf{n}$, $\sum_{j\in\mathcal{P}_i}  s_{ij}= \sum_{j\in\mathbf{n}} s_{ij}$ and
    $\sum_{j\in\mathbf{n}} s_{ij} = 1$, it follows that
    \begin{equation*}
        \sum_{j=1}^n|s_{ij} - \delta _j| = 1 -\sum_{j\in\mathcal{P}_i}\delta_j + \sum_{j\in\bar{\mathcal{P}}_i}|\delta_j| \;\;\;\; i\in\mathbf{n}
    \end{equation*}
    Hence, $\|S-\mathbf{1}\delta\|_{\infty}<1$ if and only $\sum_{j=1}^n|s_{ij} - \delta _j|<1,\;i\in\mathbf{n}$, or equivalently
    \begin{equation*}
        -\sum_{j\in\mathcal{P}_i}\delta_j +\sum_{j\in\bar{\mathcal{P}}_i}|\delta_j|\;<0,\;\;\;\;\;i\in\mathbf{n}
    \end{equation*}
    If the last inequality holds, then clearly
    \begin{equation*}
        -\sum_{j\in\mathcal{P}_i} |\delta_j| +\sum_{j\in\bar{\mathcal{P}}_i}|\delta_j|\;<0,
    \end{equation*}
    so we may assume that all entries of~$\delta$ are nonnegative. Finally, since the $ij$th entry of~$\mathbf{1}\mathbf{1}' - 2[[S]]$ is equal to $-1$ if and only if~$j \in \mathcal{P}_i$ and~$1$ otherwise, we conclude that~$\procsn{S}_\infty < 1$ if and only if 
    \begin{equation}\label{eq:cond_delta}
        (\mathbf{1}\mathbf{1}' - 2[[S]])\delta' < \mathbf{0}.
    \end{equation}
    Noting~\eqref{eq:cond_delta} holds for some~$\delta$ if and only if it also holds for~$a \delta$ with any~$a>0$ completes the proof.
\end{proof}

We now show that~$\procsn{S}_\infty \neq \coe(S)$ by means of a counterexample. Consider the stochastic matrix
\begin{equation}
    S = \frac{1}{3}\begin{bmatrix}
        1 & 0 & 0 & 1 & 0 & 1 \\
        1 & 1 & 1 & 0 & 0 & 0 \\
        0 & 0 & 1 & 0 & 1 & 1 \\
        0 & 1 & 0 & 1 & 0 & 1 \\
        1 & 0 & 0 & 1 & 1 & 0 \\
        0 & 1 & 0 & 0 & 1 & 1
    \end{bmatrix}.
\end{equation}
It is easy to verify that~$S$ is a scrambling matrix, so that $\coe(S)<1$. Suppose to obtain a contradiction that~$\procsn{S}_\infty < 1$, and note that $[[S]] = 3S$. Then, by Proposition~\ref{prop:procsn_infty_contract}, there exists a nonnegative vector~$y$ such that
\begin{equation}
    (\mathbf{1}\mathbf{1}' - 6S)y < \mathbf{0}.
\end{equation}
It then follows that for any nonnegative and nonzero~$x \in \R^n$,
\begin{equation}
    x' (\mathbf{1}\mathbf{1}' - 6S) y < 0.
\end{equation}
However, this cannot be true since~$x = \begin{bmatrix} 0 & 1 & 1 & 1 & 1 & 0 \end{bmatrix}'$ is such that~$x' (\mathbf{1}\mathbf{1}' - 6S)=\mathbf{0}$. We conclude that~$\procsn{S}_\infty = 1$ and thus that $\coe(S)<\procsn{S}_\infty$.

\section{Induced consensus seminorms}\label{sec:induced}
The idea of ``inducing'' a matrix norm on $\R^{n \times m}$ using  a vector norm on $\R^m$, suggests that a useful way to try to  define a matrix seminorm on $\R^{n \times m}$ might be to try to induce it with the vector seminorm $\procsn{\cdot}_p$ on $\R^m$.  
In other words, starting with the vector seminorm $\procsn{x}_p = \min_c \|x-\mathbf{1}c\|_p$, try to define, if possible, a continuous, nonnegative, scalar-valued function $\indsn{\cdot}_p:\R^{n \times m} \to \R$ so that for any matrix $M_{n \times m}$,
\begin{equation}\label{eq:desired_compat}
    \procsn{Mx}_p \leq \indsn{M}_p \procsn{x}_p, \;\; \forall x\in\R^m
\end{equation}
We claim that \textbf{\emph{no such function exists}}, whether it is a seminorm or not.
To understand why this is so consider first the case when $m=1$; then $x$ is a scalar and $M$ is an $n-$vector. Assume $x$ is nonzero. Even so, under these conditions $\procsn{x}_p = 0$, so \eqref{eq:desired_compat} cannot hold for any value of $\indsn{M}_p$
unless $\procsn{Mx}_p = 0$, which only occurs in the special case when all rows of $M$, i.e. all entries in its single column, are equal.
 
Consider next the case when $m>1$ and $x$ is a nonzero vector in the one-dimensional subspace $\consubsp \dfb \operatorname{span} \mathbf{1}$;
suppose that $M$ is any matrix for which $M\mathbf{1} \not\in \consubsp$. Then $\procsn{x}_p = 0$ and $\procsn{Mx}_p \neq 0$  so there is no way to define $\indsn{M}_p$ so that the inequality in \eqref{eq:desired_compat} holds for these values of $M$ and $x$.
 
There are two possible ways to deal with this issue. The first is to restrict the values of $x\in\R^m$ for which \eqref{eq:desired_compat} is required to hold to some subset $U$ whose intersection with $\consubsp$ is contained in the subspace
  $M^{-1}(\consubsp) \dfb \{x : Mx\in\consubsp\}$.
Note that if $U$ is any subspace $\mathcal{U}$ for which $\mathcal{U}\cap\mathcal{I} = 0$, then $U$ will have the required property.
Any subspace for which $\mathcal{U} \oplus \consubsp = \R^m$ would work; the orthogonal complement of $\consubsp$  is one such choice \cite{DePasquale2024seminorms}, but there is no compelling reason to choose $\mathcal{U}$ in this way.

A second and perhaps more natural way to deal with this issue is to restrict the definition of the induced seminorm to the subset $\eqrowsum^{n\times m}$ of $\R^{n\times m}$ consisting of those matrices $M$ whose row sums are all equal\footnote {Note that in the case when $m=n$, $\eqrowsum^{n \times n}$  contains all  $n\times n$ stochastic matrices.}.
In this case  $\procsn{Mx}_p=0$ whenever $\procsn{x}_p = 0$. 
For the special situation when $m=1$, both $\procsn{x}_p$ and $\procsn{Mx}_p$ are zero so in this case \eqref{eq:desired_compat} will hold for any value of $\indsn{M}_p$.
Suppose next that $m>1$. In this case it makes sense to define
\begin{equation}\label{eq:def_indsn}
    \indsn{M}_p = \max_{x\in\mathcal{X}} \procsn{Mx}_p, \;\; M\in\eqrowsum^{n\times m}
\end{equation}
where  $\mathcal{X} \dfb \{x:\procsn{x}_p=1,x\in\R^m\}$. To understand why \eqref{eq:def_indsn} implies \eqref{eq:desired_compat}, fix $x\in\R^m$.
If $x\in\consubsp$, then $\procsn{x}_p$ and $\procsn{Mx}_p$ are both zero and \eqref{eq:desired_compat} holds no matter what the value of $\indsn{M}_p$.
Meanwhile, if $x\not\in\consubsp$, then $\procsn{x}_p \neq 0$ so
there must be a number $r>0$ so that $r\procsn{x}_p = 1$ in which case
$rx\in\mathcal{X}$. Thus, appealing to \eqref{eq:def_indsn},  $\procsn{M(rx)}_p \leq \indsn{M}_p \procsn{rx}_p$  which implies that \eqref{eq:desired_compat} holds.
It is worth noting that for all $p$, $\indsn{M}_p \leq \procsn{M}_p,\;M\in\eqrowsum^{n\times m}$. This is because  $\indsn{M}_p = \procsn{Mx}_p$ for some $x$ satisfying $\procsn{x}_p = 1$ and because
the submultiplicative property  $\procsn{Mx}_p\leq \procsn{M}_p \procsn{x}_p$ holds for such $M$. Also note that it immediately follows from~\eqref{eq:desired_compat} and~\eqref{eq:def_indsn} that
\begin{equation}
    \indsn{M_2 M_1}_p \le \indsn{M_2}_p \indsn{M_1}_p
\end{equation}
for any~$M_2, M_1 \in \eqrowsum^{n \times n}$.

\subsection{Special Case: $p=\infty$}
In this subsection we will show that for $M\in\eqrowsum^{n\times n}$, the induced consensus seminorm $\indsn{M}_{\infty}$ is in fact equal to the coefficient of ergodicity of $M$. This is an immediate consequence of the following proposition.

\begin{proposition}\label{prop:indsninfty_coe}
    For $M\in\eqrowsum^{n\times n}$
    \begin{equation}\label{eq:indsninfty_coe}
        \indsn{M}_{\infty} = \frac{1}{2}\max_{i,j}\sum_{k=1}^n|m_{ik}-m_{jk}|
    \end{equation}
\end{proposition}

In the proof below use will be made of the fact that for any positive integer $k$ and any vector $x\in\R^k$
\begin{equation}\label{eq:sninfty_maxmin}
    \procsn{x}_{\infty} = \frac{1}{2}(x_{\max}-x_{\min})
\end{equation}
where $x_{\rm min}$ and $x_{\rm max}$ are respectively, the smallest and largest entries in $x$.
This easily proved fact was noted in \cite{Liu2014InternalStabilityConsensus} where it was assumed that $x$ was a nonnegative vector.
However \eqref{eq:sninfty_maxmin} actually holds for all vectors, whether they are nonnegative or not.
This is a simple consequence of the aforementioned shift invariance property of the seminorm $\procsn{\cdot}_{\infty}$.

Before proceeding with the proof of Proposition \ref{prop:indsninfty_coe}, it is worth pointing out that there is a close connection between what we are discussing here and Theorem 3.1 of \cite{Seneta2006} which contains the inequality
\begin{equation}\label{eq:sen}
    (y_{\max}-y_{\min})\leq \coe(S)(x_{\max}-x_{\min})
\end{equation}
where $S$ is any $n\times n$ stochastic matrix, $\coe(\cdot)$ is as defined in \eqref{eq:coe}, and $y_{\max}$ and $y_{\min}$ are the largest and smallest entries of the vector $y\dfb Sx$ respectively. In view of \eqref{eq:sninfty_maxmin}
it is clear that \eqref{eq:sen} is essentially the same as the inequality
$\procsn{Sx}_{\infty} \leq \indsn{S}_{\infty}\procsn{x}_{\infty}$.
It is interesting that \cite{Seneta2006} derives \eqref{eq:sen} without any reference to the seminorm  $\procsn{\cdot}_{\infty}$ or the induced seminorm $\indsn{M}_{\infty}$.

\textit{Proof of Proposition \ref{prop:indsninfty_coe}:}
Fix $M\in\eqrowsum^{n\times n}$ and suppose that $\indsn{M}_{\infty}$ is defined by \eqref{eq:def_indsn}. Since the row sums of $M$ are all equal, it must be true that $M\mathbf{1} = \sigma\mathbf{1}$ where $\sigma$ is the row sum of each row of $M$.
Let $\mathcal{\bar{X}}$ denote the set of vectors in $\R^n$ such that $x_{\max} = 1$ and $x_{\min}=-1$.
In the light  of \eqref{eq:sninfty_maxmin}
and the definition of $\mathcal{X}$ it is clear that $\bar{\mathcal{X}}$ is that subset of $\mathcal{X}$ consisting of vectors $x$ for which $x_{\max} = 1$. We claim that restricting  the admissible values of $x$ to $\bar{\mathcal{X}}$ rather than $\mathcal{X}$ does not alter the value of $\indsn{M}_{\infty}$ in \eqref{eq:def_indsn}. In other words
\begin{equation}\label{eq:altdef_indsn}
    \indsn{M}_{\infty} = \max_{x\in\bar{\mathcal{X}}}|Mx|_{\infty}
\end{equation}
To understand why this is so, suppose $x$ is any vector in $\mathcal{X}$ and define $\bar{x} = x+r\mathbf{1}$ where $r=1-x_{\max}$. Since $\procsn{x}_{\infty} = 1$, $\bar{x}\in \bar{\mathcal{X}}$.
Moreover, since $M\in\eqrowsum^{n \times n}$, $M\bar{x} = Mx +\sigma r \mathbf{1}$ where $\sigma$ is the sum of the entries in any one row of $M$. In light of the shift invariance property of $\procsn{\cdot}_{\infty}$, it follows that $\procsn{\bar{x}}_{\infty} = \procsn{x}_{\infty}$, that $\procsn{M\bar{x}}_{\infty} = \procsn{Mx}_{\infty}$, and consequently that \eqref{eq:altdef_indsn} is justified.

In view of \eqref{eq:altdef_indsn} there must be a vector $y\in\bar{\mathcal{X}}$ such that $\indsn{M}_{\infty} = \procsn{My}_{\infty}$.
Suppose that $q$ and $s$ are such that $\sum_k m_{qk}y_k$ and $\sum_k m_{sk}y_k$ are the largest and smallest entries in $My$ respectively. In view of \eqref{eq:def_indsn}
\begin{equation*}
    \indsn{M}_{\infty} = \frac{1}{2}\sum_k (m_{qk}-m_{sk}) y_k
\end{equation*}
But $y\in\bar{\mathcal{X}}$ which implies that $|y_k| \leq 1, k\in\mathbf{n}$.
Thus $(m_{qk} - m_{sk})y_k\leq |m_{qk} - m_{sk}|,\;k\in\mathbf{n}$, so $\indsn{M}_{\infty} \leq  \frac{1}{2}\sum_k|m_{qk}-m_{sk}|$. Therefore
\begin{equation}\label{aid}
    \indsn{M}_{\infty}\leq \frac{1}{2}\max_{i,j}\sum_k|m_{ik}-m_{jk}|
\end{equation}

Suppose that $u$ and $v$ are row indices such that
\begin{equation}\label{bum}
    \sum_k|m_{uk}-m_{vk}| = \max_{i,j}\sum_k|m_{ik}-m_{jk}|
\end{equation}
If $\sum_k|m_{uk}-m_{vk}| = 0$, then $\indsn{M}_{\infty}\leq 0$ because of \eqref{aid} and \eqref{bum}; therefore $\indsn{M}_{\infty} = 0$ which implies that \eqref{eq:indsninfty_coe} holds in this case.

Now suppose that $\sum_k|m_{uk}-m_{vk}| > 0$ which means that rows $u$ and $v$ are not equal.
On the other hand,
by assumption all row sums of $M$ are the same so $\sum_k(m_{uk}-m_{vk}) = 0$. This means that in the set of $n$ numbers, $\{m_{uk}-m_{vk}: k\in\mathbf{n}\}$ there must be at least one positive number and one negative number. Hence if $z$ is defined so that $z_k = \operatorname{sign} (m_{uk}-m_{vk}) $ when $(m_{uk}-m_{vk})\neq 0$ and $z_k = 0$ otherwise,
then $z\in\bar{\mathcal{X}}$
and $\sum_k(m_{uk}-m_{vk})z_k =\sum_k|m_{uk}-m_{vk}| $.
But $\sum_k(m_{uk}-m_{vk})z_k$ is the difference between $u$th and  $v$th entries of $Mz$. From this and \eqref{eq:def_indsn}
it follows that
$\procsn{Mz}_{\infty}\geq \frac{1}{2}\sum_k(m_{uk}-m_{vk})z_k$. Therefore $\procsn{Mz}_{\infty}\geq \frac{1}{2}\sum_k|m_{uk}-m_{vk}|$.
Since $z\in\bar{\mathcal{X}}$, it must be true that
$\indsn{M}_{\infty} \geq \frac{1}{2}\sum_k|m_{uk}-m_{vk}| $. Therefore $\indsn{M}_{\infty}\geq \frac{1}{2}\max_{i,j}\sum_k|m_{ik}-m_{jk}|$
because of \eqref{bum}. From this and \eqref{aid} it follows that \eqref{eq:indsninfty_coe} is true.
\hspace*{\fill}~\QED

\subsection{Contraction in induced consensus seminorms}

This section shows that contraction in the induced consensus seminorm defined above of each matrix in a product implies convergence of the product to a rank one matrix. A matrix~$M \in \eqrowsum^{n \times n}$ is called contractive in the induced seminorm~$\indsn{\cdot}_p$ if~$\indsn{M}_p < 1$.
\begin{theorem}\label{thm:induced_stochastic_contract}
    Let $p$ be fixed and let $\mathcal{C}$ be a compact subset of stochastic matrices which are all contractive in the induced seminorm $\indsn{\cdot}_p$ on $\eqrowsum^{n\times n}$. Let
    \begin{equation}\label{eq:contract}
        \lambda \dfb \max_{\mathcal{C}} \indsn{M}_p.
    \end{equation}
    Then for each infinite sequence of matrices $M_1,M_2,\ldots $ in $\mathcal{C}$, the matrix product $M_i M_{i-1} \cdots M_1$ converges as $i\to\infty$ to a rank one matrix of the form $\mathbf{1}c'$ as fast as $\lambda^i$ converges to zero.
\end{theorem}

It is worth noting that although it is fairly obvious from the hypothesis of the preceding theorem that  $|M_iM_{i-1}\cdots M_1|_p$ converges to $0$ as $i\longrightarrow\infty$, it is not so obvious that the matrix products $M_iM_{i-1}\cdots M_1$ converges to a limit as well. This in fact is why there is a need  for the preceding theorem and its proof.

\begin{proof}
Let $M_1,M_2,\ldots$ be a sequence of contractive matrices in $\mathcal{C}$ with the aforementioned properties.
To prove the theorem, it is enough to show that with $d$ an arbitrary but fixed nonzero vector,  the sequence of vectors
\begin{equation}\label{eq:seq}
    x_i = M_i M_{i-1} \cdots M_1d, \;\; i\geq 1
\end{equation}
converges to a point in the subspace spanned by the vector $\mathbf{1}$.
Towards this end, for each $i\geq 1$  let $r_i$ denote that real number $r$  which minimizes $\|x_i-r\mathbf{1}\|_p$.

From the submultiplicativity of $\indsn{\cdot}_p$ and \eqref{eq:seq}, it follows that $\procsn{x_i}_p \leq \indsn{M_i}_p\indsn{M_{i-1}}_p\cdots\indsn{M_1}_p\procsn{d}_p,i\geq 1$. This and \eqref{eq:contract} thus imply that $\procsn{x_i}_p\leq \lambda^i\procsn{d}_p,\;i\geq 1$. Therefore $\procsn{x_i}_p$ converges to zero as $i\to \infty$ as fast as $\lambda^i$ does.
Since $\procsn{x_i}_p = \|x_i - r_i\mathbf {1}\|_p$, the sequence $x_i - r_i\mathbf {1},\;\;i\geq 1$ must also converge to zero as fast as $\lambda^i$ does. To complete the proof it is enough to show that the sequence $r_i,\;i\geq 1$ converges to a finite limit $\bar{r}$ as fast as $\lambda^i$ does.

Let 
\begin{equation}\label{d1}
    e_i = x_i-\mathbf{1}r_i, \; i \geq 1
\end{equation}
then
\begin{equation}\label{invar}
    \|e_i\|_p = \procsn{x_i}_p, \;\;\; i \geq 1
\end{equation}
because of the definition of $\procsn{\cdot}_p$.
In view of \eqref{d1}, $e_{i+1} = x_{i+1} -\mathbf{1}r_{i+1}$.  From this and \eqref{eq:seq} it follows that $e_{i+1} = M_{i+1}x_i  -\mathbf{1}r_{i+1}$. Using \eqref{d1} again, there follows
  $e_{i+1} = M_{i+1}(e_i+\mathbf{1}r_i)  -\mathbf{1}r_{i+1}$. Therefore,
$$e_{i+1} = M_{i+1}e_i-\mathbf{1}(r_{i+1}-r_i),\;\;\;i\geq 1$$
because $M_{i+1}\mathbf{1} = \mathbf{1}$. By similar reasoning, it follows that
\begin{equation}
    e_{i+j} = M_{i+j}M_{i+j-1}\cdots M_{i+1}e_i-\mathbf{1}(r_{i+j}-r_i),\;\;\;i,j\geq 1\label{pp}
\end{equation}
Thus $\|\mathbf{1}(r_{i+j}-r_i)\|_p = \|M_{i+j}M_{i+j-1}\cdots M_{i+1}e_i-e_{i+j}\|_p$ so by the triangle inequality,
$$||\mathbf{1}(r_{i+j}-r_i)||_p \le ||M_{i+j}M_{i+j-1}\cdots M_{i+1}e_i||_p +||e_{i+j}||_p$$

Since both $\|\cdot\|_p$ and $\|\cdot\|_{\infty} $ are matrix norms on~$\R^{n\times n}$ there is a finite positive constant $\mu$  depending only on~$p$ and $n$ such that $\|Q\|_p\leq \mu\|Q\|_{\infty},\;\;Q\in\R^{n\times n}$. This and the fact that
$M_{i+j}M_{i+j-1}\cdots M_{i+1}$ is stochastic for all $i,j\geq 1$ imply that
$\|M_{i+j}M_{i+j-1}\cdots M_{i+1}\|_p\leq \mu\|M_{i+j}M_{i+j-1}\cdots M_{i+1}\|_{\infty}  = \mu$.

Therefore
\begin{align*}
  \|\mathbf{1}(r_{i+j}-r_i)\|_p &\leq \mu\|e_i\|_p + \|e_{i+j}\|_p \\
      &\le (\mu + \lambda^j)\lambda^i \\
      &\le (\mu + 1) \lambda^i,
\end{align*}
where the second inequality follows from~\eqref{invar}. Then, the sequence~$r_i, i \ge 1,$ is a Cauchy sequence and so it converges to a finite limit. Taking~$j \to \infty$ implies that convergence occurs as fast as~$\lambda^i$ converges to zero.
\end{proof}

Note that the only properties of the induced consensus seminorm that are used in the proof are that it is submultiplicative and that
\begin{equation}\label{eq:compat}
    \procsn{Mx}_p \le \indsn{M}_p \procsn{x}_p
\end{equation}
for any~$x \in \R^n$ and~$M \in \eqrowsum^{n \times n}$. However, the metric consensus seminorm satisfies both of the properties too, i.e., it is submultiplicative and~\eqref{eq:compat} still holds for any~$x\in\R^n$ and any~$M \in \eqrowsum^{n \times n}$ if~$\indsn{M}_p$ is replaced by~$\procsn{M}_p$. This implies the following result.
\begin{corollary}
    Theorem~\ref{thm:induced_stochastic_contract} is also true if the metric consensus seminorm $\procsn{\cdot }_p$  is used instead of the induced consensus seminorm $\indsn{\cdot}_p$.
\label{ron}\end{corollary}

\section{Contraction in general consensus seminorms}

Theorem~\ref{thm:induced_stochastic_contract} and Corollary \ref{ron} establish that contraction of all matrices in the sequence implies convergence of the matrix product to a rank one matrix at an exponential rate, under the assumption that all matrices in the sequence are stochastic matrices and ``contraction'' is with respect to either the induced or metric consensus seminorms. However,~\cite[Proposition 3]{Liu2011DetGossip} proves that if contraction is with respect to the metric consensus seminorms specifically, then exponential convergence to a rank one matrix occurs even if the matrices in the sequence are not stochastic, but rather just have row sums all equal to one. This suggests the obvious question of whether~\cite[Proposition 3]{Liu2011DetGossip} can be generalized to deal with a wider class of seminorms, and whether Theorem~\ref{thm:induced_stochastic_contract} can be generalized to deal with a wider class of matrices. This section shows that these results indeed generalize to all submultiplicative consensus seminorms. 

From here on, $\gensn{M}$ denotes a general matrix consensus seminorm as defined in Section~\ref{sec:consensus_seminorms} rather than the induced seminorm of Section~\ref{sec:induced} specifically. The proof of convergence will be stated for a specific convenient consensus seminorm, and the general result will then follow from the fact that all consensus seminorms are equivalent in the same manner that all matrix norms are equivalent, a fact which is proved next. Rather than studying these seminorms directly,~\eqref{eq:shift_invariance} will be used to identify any consensus seminorm~$\gensn{\cdot}: \R^{n \times n} \to \R$ with a matrix norm~$\|\cdot\| : \R^{(n-1) \times (n-1)} \to \R$. Let~$Q$ be a full-rank matrix with $n$ columns such that~$\ker Q = \consubsp$, and note that the rank-nullity theorem implies that~$Q \in \R^{(n-1) \times n}$. For any~$A \in \eqrowsum^{n \times n}$ there exists a unique~$\bar A \in \R^{(n-1) \times (n-1)}$ such that
\begin{equation}\label{eq:Aproj}
    QA = \bar A Q.
\end{equation}
Given any~$\bar A \in \R^{(n-1) \times (n-1)}$, it is also true that there is a matrix~$A \in \eqrowsum^{n \times n}$ such that~\eqref{eq:Aproj} holds. This follows from the fact that~$Q$ has full row rank, so it has a right inverse~$Q^{-1} \in \R^{n \times (n-1)}$ such that~$QQ^{-1} = I$, and taking~$A = Q^{-1} \bar A Q$ gives the required equality. In this case there is not a unique~$A \in \eqrowsum^{n \times n}$ for which~\eqref{eq:Aproj} holds, but if~$A_1,A_2 \in \eqrowsum^{n \times n}$ are such that
\begin{equation}
    QA_1 = QA_2 = \bar A Q,
\end{equation}
for some given~$\bar A$, then~$\image (A_1 - A_2) = \ker Q = \consubsp$. Since~$A_1 - A_2 \in \eqrowsum^{n \times n}$, and since any rank one matrix in~$\eqrowsum^{n \times n}$ must be of the form~$\mathbf{1}c$ for some~$c \in \R^{1 \times n}$, it follows that
\begin{equation}
    A_1 = A_2 + \mathbf{1}c
\end{equation}
for some such~$c$. It follows from~\eqref{eq:shift_invariance} that for any such~$A_1,A_2$, $\gensn{A_1} = \gensn{A_2}$. Therefore, the function~$\restnorm{\cdot} : \R^{(n-1) \times (n-1)} \to \R$ given by
\begin{equation}
    \restnorm{\bar A} = \gensn{A},
\end{equation}
where~$A \in \eqrowsum^{n \times n}$ is any matrix satisfying~\eqref{eq:Aproj}, is well-defined. It is easy to verify that~$\restnorm{\cdot}$ is a norm over~$\R^{(n-1) \times (n-1)}$. It is also clear that given any~$A \in \eqrowsum^{n \times n}$ and the unique~$\bar A \in \R^{(n-1) \times (n-1)}$ for which~\eqref{eq:Aproj} holds,
\begin{equation}\label{eq:N_to_sm}
    \gensn{A} = \restnorm{\bar A}.
\end{equation}

Combining~\eqref{eq:N_to_sm} with the well-known equivalence of norms over finite-dimensional vector spaces (see for example~\cite[Theorem 5.4.4]{HornJohnson2012MatrixAnalysis}) immediately leads to the following result.
\begin{lemma}\label{lem:seminorm_equiv}
    Let~$\gensn{\cdot}_a, \gensn{\cdot}_b : \eqrowsum^{n \times n} \to R$ be two consensus seminorms. There exist positive constants~$c_m,c_M$ such that
    \begin{equation}
        c_m \gensn{M}_a \le \gensn{M}_b \le c_M \gensn{M}_a
    \end{equation}
    for any~$M \in \eqrowsum^{n \times n}$.
\end{lemma}

The next result proves that contraction in any submultiplicative consensus seminorm implies consensus and gives an explicit bound on the rate at which consensus is reached.

\begin{theorem}\label{thm:contract_gensn}
    Let~$\gensn{\cdot} : \eqrowsum^{n \times n} \to \R$ be a submultiplicative consensus seminorm, and let~$\mathcal{C}$ be a compact subset of matrices in~$\eqrowsum^{n \times n}$ which are all contractive in~$\gensn{\cdot}$ and whose row sums are all equal one. Let
    \begin{equation}\label{eq:gensn_lambda}
        \lambda \dfb \max_{\mathcal{C}} \gensn{M}
    \end{equation}
    Then for each infinite sequence of matrices~$M_1,M_2,\dots$ in~$\mathcal{C}$, the matrix product~$M_i M_{i-1} \cdots M_1$ converges as~$i \to \infty$ to a rank one matrix of the form~$\mathbf{1}c$ as fast as~$\lambda^i$ converges to zero.
\end{theorem}
\vspace{0.25em}

Note that~$\lambda$ is well-defined as all seminorms over finite-dimensional spaces are continuous~\cite{Goldberg2017SeminormContinuous}. Also note that while it is obvious that the seminorm of the matrix product converges to zero, this alone does not imply that the matrix product converges to any finite limit.

\begin{proof}
    As in the proof of Theorem~\ref{thm:induced_stochastic_contract}, it suffices to show that given fixed nonzero vector~$d$, the sequence of vectors
    \begin{equation}
        x_i = M_i M_{i-1} \cdots M_1 d, \quad i \ge 1,
    \end{equation}
    converges to a point in the subspace spanned by the vector~$\mathbf{1}$.
    
    Let~$P \in \R^{n \times n}$ be any projection onto~$\consubsp$, i.e.~$P$ is any matrix such that~$P^2 = P$ and~$\image P = \consubsp = \operatorname{span} \mathbf{1}$. Clearly~$x_i = Px_i + (I - P)x_i$, so it is sufficient to show that both~$Px_i$ and~$(I-P)x_i$ converge to a finite limit as fast as~$\lambda^i$ converges to zero.
    
    We first show that~$(I-P)x_i$ converges to zero. It is easy to show that the function~$\gensn{M}_{I-P} \dfb \|(I-P)M\|_2$ is a seminorm. Furthermore, it is a \emph{consensus seminorm}, since~$(I-P)M=0$ if and only if all columns of~$M$ are in the kernel of~$I-P$, and~$\ker(I-P) = \image P = \operatorname{span}\mathbf{1}$. It then follows immediately from Lemma~\ref{lem:seminorm_equiv}, the submultiplicativity of~$\gensn{\cdot}$, and~\eqref{eq:gensn_lambda}, that
    \begin{align*}
        \|(I-P)M_i M_{i-1} \cdots M_1 \|_2 &\le c\gensn{M_i M_{i-1} \cdots M_1} \le c \lambda^i
    \end{align*}
    for some~$c \in \R$. Since
    \begin{align*}
        \|(I-P)x_i\|_2 &= \|(I-P)M_i M_{i-1} \cdots M_1 d\|_2 \\
            &\le \|(I-P)M_i M_{i-1} \cdots M_1\|_2 \|d\|_2,
    \end{align*}
    the sequence~$(I-P)x_i, \;i\ge1$ converges to zero as fast as~$\lambda^i$ converges to zero.
    
    We now turn to~$Px_i$. Since~$M_i \mathbf{1} = \mathbf{1}$ for all~$i$, it follows that~$M_i P = P$, and
    \begin{align*}
        Px_{i+1} &= P M_{i+1} x_i \\
            &= P M_{i+1} (P + I - P) x_i \\
            &= P x_i + P M_{i+1} (I - P) x_i.
    \end{align*}
    Then, for all~$i\ge1$,
    \begin{align*}
        Px_i &= Pd + \sum_{k=1}^{i} u_k,
    \end{align*}
    where~$u_k = P M_k (I-P)x_k$. It is enough to show that~$\sum_{k=1}^{i} u_k$ converges to a finite limit as fast as~$\lambda^i$ converges to zero. Since~$\mathcal{C}$ is compact and~$\|(I-P)x_i\|_2 \le c \lambda^i$, there is a positive constant~$\tilde c$ such that
    \begin{equation*}
        \|u_i\|_2 \le \tilde c \lambda^i, \quad i\ge1.
    \end{equation*}
    Then for all $i,j \ge 1$,
    \begin{align*}
        \left\| \sum_{k=1}^{i+j} u_k - \sum_{k=1}^{i} u_k \right\|_2 = \left\| \sum_{k=i+1}^{i+j} u_k \right\|_2 \le \tilde c \sum_{k=i+1}^{i+j} \lambda^k.
    \end{align*}
    Taking the sum of the geometric sequence on the right gives
    \begin{equation*}
        \left\| \sum_{k=1}^{i+j} u_k - \sum_{k=1}^{i} u_k \right\|_2 \le  c \lambda^{i + 1} \frac{1 - \lambda^{j}}{1-\lambda} < \frac{c}{1-\lambda},
    \end{equation*}
    and this proves that the sequence~$\sum_{k=1}^{i} u_k, \; i\ge 1$ is a Cauchy sequence and has a finite limit. Taking~$j \to \infty$ then proves that it also converges to its finite limit as fast as~$\lambda^i$ converges to zero.
\end{proof}

\section{Limitations of consensus seminorms}

This section focuses on the convergence rate of products of stochastic matrices. Given a compact subset~$\mathcal{S}$ of the set of stochastic matrices and a submultiplicative consensus seminorm~$\gensn{\cdot}$ such that~$\gensn{S} < 1$ for all~$S \in \mathcal{S}$, Theorem~\ref{thm:contract_gensn} shows that every product of infinitely many matrices from~$\mathcal{S}$ converges exponentially fast to a rank one matrix, and also provides a bound on the rate of convergence. It is therefore useful to find classes of contracting stochastic matrices, where contraction is with respect to some submultiplicative consensus seminorm. This is related to the well-studied question of determining when a class of stochastic matrices is such that every infinite product of matrices from the class converges to a rank one matrix~\cite{HajnalBartlett1958}.

We call a compact set~$\mathcal{S}$ of stochastic matrices an \emph{agreement set} if every infinite product of matrices from~$\mathcal{S}$ converges to a rank one matrix. There are several well-known classes of stochastic matrices every compact subset of which is an agreement set. These include the class of stochastic matrices with at least one column all of whose entries are positive (also known as matrices with a strongly rooted graph), the class of scrambling matrices, and the class of stochastic matrices with positive diagonals and a rooted graph~\cite{Ming2008ConsensusGraphical}. Another such class which includes all previous classes as a special case is the class of Sarymsakov matrices~\cite{semigroup} and its generalization~\cite{Xia2019GenSary}. Note that every stochastic matrix with a positive column has metric consensus seminorm~$\procsn{\cdot}_\infty$ less than one, and every scrambling matrix has induced consensus seminorm~$\indsn{\cdot}_\infty$ less than one, so Theorem~\ref{thm:contract_gensn} proves that every compact subset of either of these two classes is an agreement set, and tells us how to find a bound on the rate of convergence.

It is therefore perhaps natural to expect that there is some submultiplicative consensus seminorm whose value on any Sarymsakov matrix is less than one. Unfortunately, this is false, as shown by the next result. To be exact, the following theorem shows that there is no class of stochastic matrices larger than the class of scrambling matrices, in the sense of the partial order induced by set inclusion, where all matrices in the class are contracting with respect to a single submultiplcative consensus seminorm.

\begin{theorem}\label{thm:no_coe_gen}
    Let~$\gensn{\cdot}$ be an arbitrary submultiplicative consensus seminorm. If there is a stochastic non-scrambling matrix~$S_1$ such that~$\gensn{S_1} < 1$, then there exists a scrambling matrix~$S_2$ such that~$\gensn{S_2} > 1$.
\end{theorem}
\begin{proof}
    Suppose that~$S_1$ is stochastic non-scrambling matrix such that~$\gensn{S_1} < 1$. Since it is not a scrambling matrix, by~\cite[Theorem 1]{HajnalBartlett1958} there exists a stochastic matrix~$S_2$ such that~$S_1 S_2$ is not an agreement matrix, i.e.,~$(S_1 S_2)^i$ does not converge to a rank one matrix as~$i \to \infty$. It therefore follows from Theorem~\ref{thm:induced_stochastic_contract} that~$\gensn{S_1 S_2} \ge 1$ and by sub-multiplicativity that~$\gensn{S_1} \gensn{S_2} \ge 1$, so
    \begin{equation*}
        \gensn{S_2} \ge \frac{1}{\gensn{S_1}} > 1.
    \end{equation*}
    Since seminorms are continuous and since for any~$\varepsilon > 0$ there exists a stochastic matrix~$S_3$ with all entries positive entries (ensuring thereby it is a scrambling matrix) and such that~$\|S_3 - S_2\|_\infty < \varepsilon$, we conclude that there exists a scrambling matrix~$S_3$ such that~$\gensn{S_3} > 1$.
\end{proof}

\section{CONCLUSIONS}

In this paper we studied the role that seminorms can play in computing bounds on the rate of reaching consensus, and showed that, for stochastic matrices in particular, seminorms cannot be used to bound the rate of convergence even for some classes of matrices where convergence is known to occur. However, this does not mean that seminorms are not useful in this context. Indeed, stochastic matrices with positive diagonals appear in many interesting applications, and since not all scrambling matrices have positive diagonals, Theorem~\ref{thm:no_coe_gen} still allows the existence of a seminorm which is less than one for all stochastic matrices with positive diagonal and a rooted graph. Also, this leaves open the question of what are ``seminorm-like'' functions that can be used to derive bounds on the rate of convergence.







\bibliographystyle{IEEEtranS}
\bibliography{literature}

\end{document}